\documentclass[peerreview,12pt,draftclsnofoot,onecolumn]{IEEEtran}
\newtheorem{theorem}{Theorem}[section]
\newtheorem{lemma}[theorem]{Lemma}

\newtheorem{definition}{Definition}[section]

\newcommand{\beq}{\begin{equation}}
\newcommand{\eeq}{\end{equation}}
\newcommand{\beqq}{\[}
\newcommand{\eeqq}{\]}
\newcommand{\barr}{\begin{array}}
\newcommand{\earr}{\end{array}}
\newcommand{\berr}{\begin{eqnarray}}
\newcommand{\eerr}{\end{eqnarray}}
\newcommand{\berrr}{\begin{eqnarray*}}
\newcommand{\eerrr}{\end{eqnarray*}}

\begin{document}
%
\title{A Heterogeneous High Dimensional Approximate Nearest Neighbor Algorithm}
%
%
%

\author{Moshe~Dubiner
\thanks{M. Dubiner is with Google, e-mail: moshe@google.com}
\thanks{Manuscript submitted to IEEE Transactions on Information Theory on March 3, 2007.}
}

%
%

\markboth{Journal of \LaTeX\ Class Files,~Vol.~1, No.~1, March~2007}%
{Shell \MakeLowercase{\textit{et al.}}: Bare Demo of IEEEtran.cls for Journals}
%



\maketitle

\section*{Acknowledgment}

I would like to thank Phil Long and David Pablo Cohn for
reviewing rough drafts of this paper and suggesting many clarifications.
The remaining obscurity is my fault.

%
\IEEEpeerreviewmaketitle

\begin{abstract}

We consider the problem of finding high dimensional approximate nearest
neighbors. Suppose there are $d$ independent rare features, each having its own
independent statistics.
A point $x$ will have $x_{i}=0$ denote the absence of feature $i$,
and $x_{i}=1$ its existence. Sparsity means that usually $x_{i}=0$.
Distance between points is a variant of the Hamming distance.
Dimensional reduction converts the sparse heterogeneous problem into a lower
dimensional full homogeneous problem. However we will see that the converted
problem can be much harder to solve than the original problem.
Instead we suggest a direct approach. It consists of $T$ tries.
In try $t$ we rearrange the coordinates in decreasing order of
\beq (1-r_{t,i})\frac{p_{i,11}}{p_{i,01}+p_{i,10}} \ln\frac{1}{p_{i,1*}} \eeq
where $0<r_{t,i}<1$ are uniform pseudo-random numbers, and the $p's$ are
the coordinate's statistical parameters. The points are lexicographically
ordered, and each is compared to its neighbors in that order.

We analyze a generalization of this algorithm, show that it is optimal
in some class of algorithms, and estimate the necessary number of tries to
success. It is governed by an information like function, which we call
bucketing forest information. Any doubts whether it is ``information''
are dispelled by another paper, where unrestricted bucketing
information is defined.

\end{abstract}


\section{Introduction}  \label{introduction}

Suppose we have two bags of points, $X_{0}$ and $X_{1}$,
randomly distributed in a high-dimensional space. The points are independent of
each other, with one exception: there is one unknown point $x_{0}$
in bag $X_{0}$ that is significantly closer to an unknown point $x_{1}$
in bag $X_{1}$ than would be accounted for by chance. We want an efficient
algorithm for quickly finding these two 'paired' points.

The reader might wonder why we need two sets, instead of working as
usual with $X=X_{0}\cup X_{1}$. We have come a full circle on this issue.
The practical problem that got us interested in this theory involved
texts from two languages, hence two different sets. However it seemed that
the asymmetry between $X_{0}$ and $X_{1}$ was not important, so we developed
a one set theory. Than we found out that keeping $X_{0},X_{1}$ separate
makes thing clearer.

Let us start with the well known simple homogeneous marginally Bernoulli(1/2)
example. Suppose $X_{0},X_{1}\subset \{0,1\}^{d}$ of sizes $n_{0},n_{1}$
respectively are randomly chosen as independent Bernoulli(1/2) variables,
with one exception. Choose randomly one point $x_{0}\in X_{0}$,
xor it with a random Bernoulli($p$) vector and overwrite one randomly chosen
$x_{1}\in X_{1}$. A symmetric description is to say that
$x_{0},x_{1}$ $i$'th bits have the joint probability
\beq P=\left( \barr{cc}
p/2 & (1-p)/2 \\
(1-p)/2 & p/2
\earr \right) \eeq
For some $p>1/2$.
We assume that we know $p$. In practice it will have to be estimated.

Let
\beq \ln M=\ln n_{0}+\ln n_{1}-I(P)d \label{info}\eeq
where
\beq I(P)=p\ln(2p)+(1-p)\ln(2(1-p)) \eeq
is the mutual information between the special pair's single coordinate values.
Information theory tells us that we can not hope to pin the special pair down
into less than $W$ possibilities, but can come close to it in some asymptotic
sense. Assume that $W$ is small. How can we find the closest pair?
The trivial way to do it is to compare all the $n_{0}n_{1}$ pairs.
A better way has been known for a long time. The earliest references I am
aware of are Karp,Waarts and Zweig \cite{KWZ95}, Broder \cite{Bro98},
Indyk and Motwani \cite{IM98}. They do not
limit themselves to this simplistic problem, but their approach clearly
handles it. Without restricting generality let $n_{0}\le n_{1}$.
We randomly choose
\beq k \approx \log_{2}n_{0} \eeq
out of the $d$ coordinates, and compare the point pairs which agree on
these coordinates (in other words, fall into the same bucket).
The expected number of comparisons is
\beq n_{0}n_{1}2^{-k}\approx n_{1} \eeq
while the probability of success of one comparison is $p^{k}$.
In case of failure we try again, with other random k coordinates. At first
glance it might seem that the expected number of tries until success
is $p^{-k}$, but that is not true because the attempts are interdependent.
The correct computation is done in the next section. In the unlimited data
case $d\rightarrow\infty$ indeed
\beq T\approx p^{-k}\approx n_{0}^{\log_{2}1/p} \eeq
Is this optimal? Alon \cite{NA} has suggested the possibility of improvement
by using Hamming's perfect code. We have found that in the $n_{0}=n_{1}=n$
case, $T\approx n^{\log_{2}1/p}$ can be reduced to
\beq T\approx n^{1/p-1+\epsilon} \eeq
for any $\epsilon>0$, see \cite{MD2}.
Unfortunately this seems hard to convert into a practical algorithm.

In practice, most approximate nearest neighbor problems are heterogeneous.
Coordinates are not independent either, but there is a lot to learn from
the independent case. For starters let the joint probability matrix be
position dependent:
\beq P_{i}=\left( \barr{cc}
p_{i}/2 & (1-p_{i})/2 \\
(1-p_{i})/2 & p_{i}/2
\earr \right) \qquad 1\le i\le d\eeq
This is an important example which we will refer to as the marginally
Bernoulli(1/2) example. It turns out that in each try
coordinate $i$ should be chosen with probability
\beq \max\left[\frac{p_{i}-p_{cut}}{1-p_{cut}},0\right] \eeq
for some $\mathbf{cutoff}$ $\mathbf{probability}$ $0\le p_{cut}\le 1$.
An intuitive argument leading to that equation appears in section
\ref{intuitive}.

Section \ref{general} presents an independent data model, and a 
general nearest neighbor algorithm using its parameters.
Section \ref{lower} proves a lower bound for its success probability.
Section \ref{semiupper} proves an upper bound for a much lager class
of algorithms. The lower and upper bound are asymptotically similar.
The number of tries T satisfies
\beq \ln T\sim \max_{\lambda\ge 0}\left[\lambda\ln n_{0}-
\sum_{i=1}^{d}F(P_{i},\lambda)\right] \eeq
where $F(P_{i},\lambda)$ is defined in (\ref{up}).
The similarity to the information theoretic (\ref{info}) suggests that
$F(P_{i},\lambda)$ is some sort of information function. We call it the
$\mathbf{bucketing}$ $\mathbf{forest}$ $\mathbf{information}$ function.
\cite{MD2} proves a similar estimate for the performance of the ``best
possible'' bucketing algorithm, involving a $\mathbf{bucketing}$
$\mathbf{information}$ function with a very information theoretic look.

Section \ref{sparsity} shows that our algorithm preserves sparseness.
Section \ref{downside} shows that dimensional reduction is bad for sparse
data.

\section{The Homogeneous Marginally Bernoulli(1/2) Example}

The well known homogeneous marginally Bernoulli(1/2) example has been presented
in the introduction. It will be analyzed in detail because the main purpose
of this paper is generalizing it. The analysis is non-generalizable,
but the issues remain. Recall that we have a joint probability matrix
\beq P=\left( \barr{cc}
p/2 & (1-p)/2 \\
(1-p)/2 & p/2
\earr \right) \eeq
For some $p>1/2$. Without restricting generality let $n_{0}\le n_{1}$.
We randomly choose
\beq k\approx\min[\log_{2}n_{0}, (2p-1)d] \eeq
out of the $d$ coordinates. The reason for $k\le (2p-1)d$ (which was omitted
in the introduction for simplicity) will emerge later.
We compare point pairs which agree on the chosen $k$ coordinates.
This is a random algorithm solving a random problem, so
we have two levels of randomness. Usually when we will compute probabilities
or expectations it will be with respect to these two sources together.
The expected number of comparisons is $n_{0}n_{1}2^{-k}$
while the probability of success of one comparison is $p^{k}$.
(These statements are true assuming only model randomness).
In case of failure we try again, with other random k coordinates.
In order to estimate the expected number of tries till success we have to
enumerate how many bits are identical in the special pair $x_{0},x_{1}$.
Let this number be $j$. Then the probability of success in a single try
conditioned on $j$ is
$\left(\barr{cc}j\\k\earr\right)\left/ \left(\barr{cc}d\\k\earr\right)\right.$.
Hence the expected number of comparisons is $Tn_{1}$ where
\beqq T=n_{0}2^{-k}\sum_{j=0}^{d}\left(\barr{cc}d \\ j\earr\right)
p^{j}(1-p)^{d-j} \left(\barr{cc}d\\k\earr\right)
\left/\left(\barr{cc}j\\k\earr\right)\right.\eeqq
For small $d/k$ this is too pessimistic because most of the contribution to
the above sum comes from unlikely low $j$'s. We know that with probability
about 1/2, $j\ge pd$. Hence we get a success probability of about 1/2 with
an expected
\berrr T=n_{0}2^{-k}\sum_{j=pd}^{d}\left(\barr{cc}d \\
j\earr\right)p^{j}(1-p)^{d-j}
\left(\barr{cc}d\\k\earr\right)\left/\left(\barr{cc}j\\k\earr\right)\right.
\approx \\ \approx n_{0}2^{-k}
\left(\barr{cc}d\\k\earr\right)\left/\left(\barr{cc}pd\\k\earr\right)\right.=
n_{0}\prod_{i=0}^{k-1}\frac{1-i/d}{2(p-i/d)} \eerrr
Now it is clear that increasing $k$ above $(2p-1)d$ increases $T$, which
is counterproductive.

\section{An Intuitive Argument for the Marginally Bernoulli(1/2) Example}
\label{intuitive}
In full generality our algorithm is not very intuitive. In this section we will
present an intuitive argument for the special case of the joint probability
matrices
\beq P_{i}=\left( \barr{cc}
p_{i}/2 & (1-p_{i})/2 \\
(1-p_{i})/2 & p_{i}/2
\earr \right) \qquad 1\le i\le d\eeq
The impatient reader may skip this and the next section,
jumping directly to the algorithm.
Let us order the coordinates in decreasing order of importance
\beq p_{1}\ge p_{2}\ge \cdots \ge p_{d} \eeq
Moreover let us bunch coordinates together into $g$ groups of
$d_{1},d_{2},\dots,d_{g}$ coordinates, where $\sum_{h=1}^{g}d_{h}=d$,
and the members of group $h$ all have the same probability $q_{h}$
\beq p_{d_{1}+\cdots+d_{h-1}+1}=\cdots =p_{d_{1}+\cdots+d_{h}}=q_{h} \eeq
Out of the $d_{h}$ coordinates in group $h$, the special pair will agree in
approximately $q_{h}d_{h}$ 'good' coordinates. Let us make things simple by
pretending that this is the exact value (never mind that it is not an integer).
We want to choose
\beq k=\log_{2}n_{0} \eeq
coordinates and compare pairs which agree on them. The greedy approach seems to
choose as many as possible from the group 1, but conditional greed disagrees.
Let us pick the first coordinate randomly from group 1. If it is bad, the whole
try is lost. If it is good, group 1 is reduced to size $d_{1}-1$, out of which
$q_{1}d_{1}-1$ are good. Hence the probability that a remaining coordinate
is good is reduced to
\beq \frac{q_{1}d_{1}-1}{d_{1}-1} \eeq
After taking $m$ coordinates out of group 1, its probability decreases to
\beq\frac{q_{1}d_{1}-m}{d_{1}-m} \eeq
Hence after taking
\beq m=\frac{q_{1}-q_{2}}{1-q_{2}}d_{1} \eeq
coordinates, group 1 merges with group 2. We will randomly chose coordinates
from this merged group till its probability drops to $q_{3}$.
At that point the probability of a second group coordinate to be chosen is
\beq \frac{q_{2}-q_{3}}{1-q_{3}} \eeq
while the probability of a first group coordinate being picked either before
or after the union is
\beq \frac{q_{1}-q_{2}}{1-q_{2}}+\left(1-\frac{q_{1}-q_{2}}{1-q_{2}}\right)
\frac{q_{2}-q_{3}}{1-q_{3}} = \frac{q_{1}-q_{3}}{1-q_{3}} \eeq
This goes on till at some $q_{l} = p_{cut}$ we have $k$ coordinates.
Then the probability that coordinate $i$ is chosen is
\beq \max\left[\frac{p_{i}-p_{cut}}{1-p_{cut}},0\right] \eeq
as stated in the introduction. The cutoff probability is determined by
\beq \sum_{i=1}^{d}\max\left[\frac{p_{i}-p_{cut}}{1-p_{cut}},0\right]
\approx k \eeq

The previous equation can be iteratively solved. However it is better
to look from a different angle. For each try we will have to generate
$d$ independent uniform $[0,1]$ random real numbers
\beq 0< r_{1},r_{2},\ldots,r_{d}<1 \eeq
one random number per coordinate. Then we take coordinate $i$ iff
\beq r_{i}\le\frac{p_{i}-p_{cut}}{1-p_{cut}} \eeq
Let us reverse direction. Generate $r_{i}$ first, and then compute for which
$p_{cut}$'s coordinate $i$ is taken:
\beq p_{cut}\le 2^{-\lambda_{i}}=\max\left[\frac{p_{i}-r_{i}}{1-r_{i}},0
\right] \label{cuty}\eeq
Denoting the right hand side by $2^{-\lambda_{i}}$ is
unnecessarily cumbersome at this stage, but will make sense later.
We will call $\lambda_{i}$ the $\mathbf{random}$ $\mathbf{exponent}$ of
 coordinate $i$ (random because it is $r_{i}$ dependent).
Remember that $p_{cut}>0$ so $\lambda_{i}=\infty$  means that for that value
of $r_{i}$ coordinate $i$ can not be used.
Now which value of $p_{cut}$ will get us $k$ coordinates? There is no need
to solve equations. Sort the $\lambda_{i}$'s in nondecreasing order, and
pick out the first $k$. Hence
\beq p_{cut}=2^{-\lambda_{cut}} \eeq
where the $\mathbf{cutoff}$ $\mathbf{exponent}$ $\lambda_{cut}$ is the value
of the $k'th$ ordered random exponent.

It takes some time to comprehend the effect of equation (\ref{cuty}).
The random element seems overwhelming.
The probability that coordinate 1 will have larger
random exponent than coordinate 2 when $p_{1}>p_{2}$ is
\beq \frac{1}{2}\frac{1-p_{1}}{1-p_{2}} \eeq
In particular the probability that a useless coordinate with $p_{i}=0.5$
precedes a good coordinate with $p_{i}=0.9$ is 0.1 ! However the chance
that the useless coordinate  will be ranked among the first $k$ is very small,
unless we have so little data that it is better to take $k<\ln n_{0}$.

\section{An Unlimited Homogeneous Data Example} \label{unlimited}

The previous section completely avoids an important aspect of the general
problem which will be presented by the following example. Suppose we have
an unlimited amount of data $d\rightarrow \infty$ of the same type
\beq P=\left( \barr{cc}
p_{00} & p_{01} \\
p_{10} & p_{11}
\earr \right) \qquad 1\le i\le d\eeq
where
\beq p_{00}+p_{01}+p_{10}+p_{11}=1 \eeq
This is the joint probability of the dependent pair, and the marginal
probabilities govern the distribution of the remaining points. In the set
$X_0$ the probability that bit $i$ is 0 is
\beq p_{0*}=p_{00}+p_{01} \eeq
and similarly in $X_{1}$
\beq p_{*0}=p_{00}+p_{10} \eeq
The * means ``don't care''.
A reasonable pairing algorithm (very similar in this case to the general
algorithm) is to pick coordinates at random $1\le i_{1},i_{2},\ldots \le d$.
After picking $k$ coordinates, an $X_{0}$ point
$x_{l}=(x_{l1},x_{l2},\ldots,x_{ld})$ is in a bucket of expected size
\beq n_{0}\prod_{t=1}^{k}p_{x_{lt}*} \eeq
Hence it makes sense to increase $k$ only up to the point where 
$ n_{0}\prod_{t=1}^{k}p_{x_{lt}*} < 1$, and then compare with all $X_{1}$
points in its cell. This makes $k$ point dependent.
The expected number of comparisons in a single try is at most
$n_{1}$.  What is the approximate success probability?

Our initial estimate was the following. The probability that the special pair
will agree in a single coordinate is $p_{00}+p_{11}$
The amount of information in a single $X_{0}$ coordinate is
$ -p_{0*}\ln p_{0*}-p_{1*}\ln p_{1*} $ so we will need about
\beq k\approx\frac{\ln n_{0}}{-p_{0*}\ln p_{0*}-p_{1*}\ln p_{1*}} \eeq
coordinates, and the success probability is estimated by
\beq (p_{00}+p_{11})^k \approx n_{0}^{-\frac{\ln(p_{00}+p_{11})}
{p_{0*}\ln p_{0*}+p_{1*}\ln p_{1*}}}\eeq
This estimate turns out to be disastrously wrong. For the bad matrix
\beq \left( \barr{cc}
1-2\epsilon & \epsilon \\
\epsilon & 0
\earr \right) \eeq
with small $\epsilon$ it suggests exponent $-1/\ln\epsilon$, while clearly
it is worse than 1. The interested reader might pause to figure out what went
wrong, and how this argument can be salvaged.

There is an almost exact simple answer with a surprising flavor.
We expect $n_{0}^{-\lambda}$, so let us check that for consistency.
Pick the first coordinate. With probability $p_{00}$, the expectation $n_{0}$
is reduced to $n_{0}p_{0*}$. With probability $p_{11}$
it is reduced to $n_{0}p_{1*}$, and with probability $p_{22}=1-p_{00}-p_{11}$
the try is already lost. Hence
\beq n_{0}^{-\lambda}\approx p_{00}(n_{0}p_{0*})^{-\lambda}+
p_{11}(n_{0}p_{1*})^{-\lambda} \eeq
Happily $n_{0}$ drops out, leaving us with
\beq p_{00}p_{0*}^{-\lambda}+p_{11}p_{1*}^{-\lambda}=1 \eeq
which determines the exponent $\lambda$.
It is very easy to convert this informal argument into a formal
theorem and proof. A harder task awaits us.

\section{The General Algorithm and its Performance} \label{general}

\begin{definition}
The independent data model is the following. We generalize from bits to $b$
discrete values. Let the sets
\beq X_{0},X_{1}\subset \{0,1,\ldots,b-1\}^{d} \eeq
of cardinalities
\beq \#X_{0}=n_{0},\quad \#X_{1}=n_{1} \eeq
be randomly constructed in the following way. The $X_{0}$ points are
identically distributed independent Bernoulli random vectors, with $p_{i,j*}$
denoting the probability that coordinate $i$ has value $j$.
There is a special pair of $X_{0},X_{1}$ points, randomly chosen out of the
$n_{0}n_{1}$ possibilities. For that pair the probability that both their
$i$'th coordinates equal $j$ is $p_{i,j}$ with no dependency between
coordinates. The rest of the $X_{1}$ points can be anything.
(We abbreviate the usual notation $p_{i,jj}$ to $p_{i,j}$,
because we will consider only the diagonal and the marginal probabilities.)
Denote
\beq p_{i,b}=1-\sum_{j=0}^{b-1}p_{i,j} \eeq
\beq P_{i}=\left(\barr{llll}p_{i,0}&p_{i,1}&\ldots&p_{i,b-1}\\
p_{i,0*}&p_{i,1*}&\ldots&p_{i,b-1\ *}\earr\right) \eeq
\end{definition}
We propose the following algorithm. It consists of several bucketing tries.
For each try we generate $d$ independent uniform $[0,1]$ random real numbers
\beq 0< r_{1},r_{2},\ldots,r_{d}<1 \eeq
one random number per coordinate.
For each coordinate $i$ we define its $\mathbf{random}$ $\mathbf{exponent}$
$\lambda_{i}\ge 0$ to be the unique solution of the monotone equation
\beq \sum_{j=0}^{b-1}\frac{p_{i,j}}{
(1-r_{i})p_{i,j*}^{\lambda_{i}}+r_{i}}=1 \label{imp}\eeq
or $+\infty$ when there is no solution.
($p_{i,j*}^{\lambda_{i}}$ means $(p_{i,j*})^{\lambda_{i}}$).
We lexicographically sort all the $n_{0}+n_{1}$ points, with lower exponent
coordinates given precedence over larger exponent coordinates, and the
coordinate values $0,1,\ldots,b-1$ arbitrarily arranged, even without
consistency. Now each $X_{1}$ point is compared with the preceding $a$ and
following $a$ $X_{0}$ points (or fewer near the ends). The comparisons are
done in some one-on-one way, and the algorithm is considered successful if it
asks for the correct comparison. The best $a$ is problem and computer
dependent, but is never large. Each try makes at most $2an_{1}$ comparisons.
Of course there is extra $n_{0}+n_{1}$ point handling work.

A nice way to write the lexicographic ordering of the algorithm follows.
Suppose that in try $t$ the sorted random exponents are
\beq \lambda_{\pi_{1}}<\lambda_{\pi_{2}}<\cdots<\lambda_{\pi_{d}} \eeq
Then each point
\beq x=(x_{1},x_{2},\ldots,x_{d})\in \{0,1,\ldots ,b-1\}^{d} \eeq
is projected into the interval $[0,1]$ by
\beqq R_{t}(x)=\sum_{i=1}^{d}p_{\pi_{1},x_{\pi_{1}}*}\ p_{\pi_{2},x_{\pi_{2}}*}
\ \cdots\ p_{\pi_{i-1},x_{\pi_{i-1}}*}
\sum_{j=0}^{x_{\pi_{i}}-1}{p_{\pi_{i},j*}} \eeqq
The projection order is a lexicographic order. For large dimension $d$,
$R_{t}(x)$ is approximately uniformly distributed in $[0,1]$.

We will prove that the number of tries $T$ needed for success satisfies
\beq \ln T\sim \max_{\lambda\ge 0}\left[\lambda\ln n_{0}-
\sum_{i=1}^{d}F(P_{i},\lambda)\right] \label{cut} \eeq
where
\begin{definition}
The bucketing forest information function $F(P_{i},\lambda)$ is
\berr F(P_{i},\lambda)=\min_{\barr{ccc}0\le q_{i,0},\ldots,q_{i,b}\\
\sum_{j=0}^{b}q_{i,j}=1\\ \sum_{j=0}^{b-1}\frac{q_{i,j}}{p_{i,j*}^{\lambda}}
\le 1\earr} \sum_{j=0}^{b}p_{i,j}\ln\frac{p_{i,j}}{q_{i,j}}=\label{up}\\=
\max_{0\le r_{i}\le 1}\sum_{j=0}^{b}p_{i,j}\ln\left(1-r_{i}+
r_{i}\frac{(j\ne b)}{p_{i,j*}^{\lambda}}\right)\label{low} \eerr
The two dual extrema points are related by
\beq q_{i,j}=\frac{p_{i,j}}{1-r_{i}+r_{i}\frac{(j\ne b)}{p_{i,j*}^{\lambda}}}
\eeq
For $\sum_{j=0}^{b-1}\frac{p_{i,j}}{p_{i,j*}^{\lambda}}\le 1$
$r_{i}=0,\ q_{i,j}=p_{i,j},\ F(P_{i},\lambda)=0$. Otherwise
\beq \sum_{j=0}^{b-1}\frac{q_{i,j}}{p_{i,j*}^{\lambda}}=
\sum_{j=0}^{b-1}\frac{p_{i,j}}{(1-r_{i})p_{i,j*}^{\lambda}+r_{i}}=1 \eeq
\end{definition}
We will get (\ref{up}) from the upper bound theorem, and (\ref{low}) from
the lower bound theorem. Their equivalence is a simple (though a bit
surprising) application of Lagrange multipliers in a convex setting.
Representation (\ref{low}) implies that $F(P,\lambda)$ is an increasing
convex function of $\lambda$.

The $\mathbf{cutoff}$ $\mathbf{exponent}$ $\lambda_{cut}$ attains (\ref{cut}).
It has several meanings.
\begin{enumerate}
\item
In each try the coordinates with $\lambda_{i}\le \lambda_{cut}$ define
a bucket of size $e^{\epsilon n_{0}}$ for some small real $\epsilon$.
\item
If we double $n_{0}$ the number of tries needed to achieve success probability
$1/2$ is approximately multiplied by $2^{\lambda_{cut}}$.
\item
If we delete coordinate $i$, then the number of tries needed to achieve
success probability $1/2$ is on average multiplied by
$e^{F_{i,p}(\lambda_{cut})}$.
\end{enumerate}

Switching $X_{0}$ and $X_{1}$ may result in a different algorithm.
Coordinate values can be changed and/or merged in possibly
different ways for $X_{0},X_{1}$. For each possibility we have an estimate of
its effectiveness, and the best should be taken.

In real applications there is dependence, and the probabilities have to be
estimated. Our practical experience indicates that this is a robust algorithm.
Details will be described in another paper.

\section{An Alternative Algorithm} \label{alternative}

There is an interesting alternative to the random ordering of coordinates.
Suppose we have training sets $X_{0},X_{1}$ both of size $n$, such that each
$X_{0}$ point is paired with a known $X_{1}$ point. Let us estimate
the probabilities $P_{i}$ by their empirical averages.
For each coordinate $i$ its $\mathbf{exponent}$ $\lambda_{i}\ge 0$ is defined
by
\beq \sum_{j=0}^{b-1}\frac{p_{i,j}}{p_{i,j*}^{\lambda_{i}}}=1 \eeq
Arrange the coordinates in the $\mathbf{greedy}$ $\mathbf{order}$ of
nondecreasing exponents. Perform the first try using that order just like in
the previous algorithm. Remove the pairs found from the training data, and
repeat recursively on the reduced training data. Stop after the training set is
reduced to $1/3$ (for example) of its original size, or you run out of memory.
The memory problem can be alleviated by keeping only the heads of
coordinate lists, and/or running training and working tries in parallel.

This simpler algorithm has a more complicated and/or less efficient
implementation, and lacks theory.

\section{Return of the Marginally Bernoulli(1/2) Example} \label{return}

For the marginally Bernoulli(1/2) example equation (\ref{imp}) is
\beq 2\frac{p_{i}/2}{(1-r_{i})2^{-\lambda_{i}}+r_{i}}=1 \eeq
which can be recast as the familiar
\beq 2^{-\lambda_{i}} = \frac{p_{i}-r_{i}}{1-r_{i}} \eeq
The bucketing forest information function is
\beqq F(P_{i},\lambda)=\left\{\barr{ll}p_{i}\ln\frac{p_{i}}{2^{-\lambda}}+
(1-p_{i})\ln\frac{1-p_{i}}{1-2^{-\lambda}}
 & p_{i}\ge 2^{-\lambda} \\
 0 & p_{i}\le 2^{-\lambda} \earr\right. \eeqq
The cutoff exponent attains (\ref{cut}). The extremal condition is the familiar
\beq \sum_{i=1}^{d}\max\left[\frac{p_{i}-2^{-\lambda_{cut}}}
{1-2^{-\lambda_{cut}}},0\right]=\log_{2}n_{0} \eeq

Let us now specialize to $p_{1}=p_{2}=\cdots=p_{d}=p$. Then
\beq \lambda_{cut}=-\log_{2}\frac{pd-\log_{2}n_{0}}{d-\log_{2}n_{0}} \eeq
Notice that $\log_{2}n_{0}>(2p-1)d$ is equivalent to $\lambda_{cut}>1$.
In general $\lambda_{cut}>1$ signals that the available bucketing forest
information is of such low quality that the trees are worse than random near
their leafs.

\section{Sparsity} \label{sparsity}

Let us specialize to sparse bits: $b=2$,
\beq p_{i,1*},p_{i,}+p_{i,11}<<1 \label{asy1}\eeq
We will also assume that for some fixed $\delta>0$
\beq p_{i,11}\ge\delta\left(p_{i,01}+p_{i,10}\right) \label{asy2}\eeq
The equation
\beq \frac{p_{i,00}}{(1-r_{i})p_{i,0*}^{\lambda_{i}}+r_{i}}+
\frac{p_{i,11}}{(1-r_{i})p_{i,1*}^{\lambda_{i}}+r_{i}} =1 \eeq
has two asymptotic regimes: one in which $p_{i,0*}^{\lambda_{i}}$ is nearly
constant and $p_{i,1*}^{\lambda_{i}}$ changes, and vice versa. The first
regime is the important one:
\beq p_{i,00}+\frac{p_{i,11}}{(1-r_{i})p_{i,1*}^{\lambda_{i}}+r_{i}}
\approx 1\eeq
\beq \lambda_{i}\approx \frac{\ln\left[1-\frac{1}
{(1-r_{i})\left(1+\frac{p_{i,11}}{p_{i,01}+p_{i,10}}\right)}\right]}
{\ln p_{i,1*}} \label{sparse}\eeq
In practice the probabilities have to be estimated from the data, and sparse
estimates must be unreliable, so we used the more conservative
\beq 1/\tilde{\lambda}_{i}=(1-r_{i})\frac{p_{i,11}}{p_{i,01}+p_{i,10}}
\ln\frac{1}{p_{i,1*}} \label{cons}\eeq

A very important practical point is that the general algorithm 
preserves sparsity. Suppose that instead of points
\beq x=(x_{1},x_{2},\ldots,x_{d})\in \{0,1\}^{d} \eeq
we have subsets of a features set $D$ of cardinality $d$ :
\beq D_{x}\subset D \eeq
In try $t$ we use a hash function ${\rm hash}_{t}:D\rightarrow [0,1]$.
For each feature $i\in D$ its random exponent $\lambda_{i}$ is computed
using the pseudo random
\beq  r_{i}={\rm hash}_{t}(i) \eeq
and the random exponents of $x$ are sorted
\beq \lambda_{\pi_{1}}<\lambda_{\pi_{2}}<\cdots<\lambda_{\pi_{\nu}} \eeq
Then the sequence of features
\beq (\pi_{1},\pi_{2},\dots,\pi_{\nu}) \eeq
is a sparse representation of $x$ whose lexicographic order is used in try $t$.

\section{The Downside of Dimensionality Reduction} \label{downside}

Another way of handling sparse approximate neighbor problems is to convert
them into dense problems by a random projection. For dense problems taking
some $k$ out of the $d$ coordinates can be an effective way to reduce
dimension. For sparse problems such a sampling reduction will remain sparse,
hence dense projection matrices are used instead. We will show that this can
result in a much worse algorithm. Let us consider the unlimited homogeneous
data example with
\beq p_{01}=p_{10},\quad n_{0}=n_{1}=n \eeq
because in general it is
not clear which projections to take and how to analyze their performance.
We have a $d$ dimensional Hamming cube $\{0,1\}^{d}$. The Hamming distance
between two random $X_{0},X_{1}$ points is approximately
\beq 2p_{0*}p_{1*}d \eeq
The Hamming distance between the two special points is approximately
\beq 2p_{0,1}d \eeq
Hence when the dimension $d$ is large, the random to special distances ratio
tends to
\beq c=\frac{p_{0*}p_{1*}}{p_{01}} \eeq
The ideal dimensionality reduction would be to project $\{0,1\}^{d}$
into a much lower dimensional $\{0,1\}^{k}$ in such a way that the images
of the $X_{0},X_{1}$ points are random $\{0,1\}^{k}$ points, and the
distance between the two special images is approximately $k/2c$
($k/2$ is the approximate distance between two random image points).
Hence after the dimensionality reduction we will have a homogeneous marginally
Bernoulli(1/2) problem with
\beq p=1-1/2c \eeq
The standard nearest neighbor algorithm solves this in approximately
\beq n^{\log_{2}\frac{2c}{2c-1}} \eeq
tries. Actual dimensional reductions fall short of this ideal.
The Indyk and Motwani theory \cite{IM98} states that
\beq n^{1/c} \eeq
tries suffice. The truth is somewhere in between.

In contrast without dimensionality reduction our algorithm takes
approximately $n^{\lambda}$ tries where $\lambda$ is determined by
\beq \frac{1-p_{1*}-p_{01}}{(1-p_{1*})^{\lambda}}+
\frac{p_{11}}{p_{1*}^{\lambda}}=1 \label{spaeq}\eeq
In the asymptotic region (\ref{asy1},\ref{asy2}) inserting $r=0$ into
(\ref{sparse}) results in
\beq \lambda \approx \frac{\ln\left[1+\frac{2p_{01}}{p_{11}}\right]}
{\ln 1/p_{1*}} \approx \frac{\ln\frac{c+1}{c-1}}{\ln 1/p_{1*}} \eeq
We encourage the interested reader to look at his favorite dimensional
reduction scheme, and see that the $\ln 1/p_{1*}$ factor is really lost.

\section{Lexicographic and Bucketing Forests} \label{forest}

Our general algorithm is of the following type.
\begin{definition}
A lexicographic tree algorithm is the following.
The $d$ coordinates are arranged according to some permutation.
Than a complete lexicographic ordered tree is
generated. It is defined recursively as a root pointing towards $b$ subtrees,
with the edges denoting the possible values of the first (after permutation)
coordinate arbitrarily ordered. The subtrees are complete lexicographic
ordered trees for the remaining $d-1$ coordinates. In particular the
lexicographic tree has $b^{d}$ ordered leafs, each denoting a point in
$\{0,1,\ldots,b-1\}^{d}$.
A lexicographic tree  algorithm arranges the $n_{0}+n_{1}$ $X_{0}\cup X_{1}$
points according to the tree, and then compares each $x_{1}$ point with its $a$
neighbors right and left. This insures no more than $2an_{1}$ comparisons
per tree.
A lexicographic forest is simply a forest of lexicographic trees,
each having its own permutation. It succeeds iff at least one tree succeeds.
\end{definition}
An obvious generalization is
\begin{definition} A semi-lexicographic tree algorithm has a 'first'
coordinate and then recursively each subtree is semi-lexicographic,
until all coordinates are exhausted.
\end{definition}
For example we can start with coordinate 3, and than consider coordinate 5
if the value is 0, or coordinate 2 if the value is 1 and so on.

The success probability of a lexicographic forest is very complicated,even
before randomizing the algorithm. For that reason we will consider an uglier
non-robust class of algorithms that are easier to understand and analyze.
\begin{definition}
A bucketing tree algorithm is predictably recursively defined.
Either compare all pairs (a leaf bucket), or take one coordinate,
split the data into $b$ parts according to its value (some parts
may be empty), and apply a bucketing tree algorithm on each part separately.
In order to have no more than $an_{0}$ expected comparisons we will insist
that each leaf expects no more than $a$ points belonging to $X_{0}$.
A bucketing forest is simply a forest of bucketing trees.
It succeeds iff at least one tree succeeds.
\end{definition}

The success probability of a bucketing forest is no bed of roses.
Let us denote a leaf by $w\in\{0,1,\ldots,b\}^{d}$, with b indicating that the
corresponding coordinate is not taken. The leaf $w$ expects
\beq n_{0}\prod_{i=1}^{d}\left\{\barr{ll}p_{i,w_{i}*} & w_{i}<b \\
1 & w_{i}=b \earr \right. \eeq
$X_{0}$ points, and its success probability is
\beq \prod_{i=1}^{d}\left\{\barr{ll}p_{i,w_{i}} & w_{i}<b \\
1 & w_{i}=b \earr \right. \eeq
The success probability of a tree is the sum of the success probabilities
of its leafs. The success probability of the whole forest is less than
the tree sum. Suppose the whole forest contains $L$ leafs
$w_{1},w_{2},\ldots,w_{L}$.
 Let $y\in\{0,1,\ldots,b\}^{d}$ denote the abbreviated state
of the special points:
\beq y_{i} = \left\{ \barr{ll}x_{0,i} & x_{0,i}=x_{1,i} \\
b & x_{0,i}\ne x_{1,i} \earr \right. \eeq
The value $b$ denotes disagreement and its probability is
$ p_{i,b}=1-\sum_{j=0}^{b-1}p_{i,j} $.
The success probability of the whole forest is
\berrr S=\sum_{y\in\{0,1,\ldots,b\}^{d}}\prod_{i=1}^{d}p_{i,y_{i}}\cdot
\qquad\qquad\qquad\qquad\qquad \\ \qquad\qquad
\cdot\left[1-\prod_{l=1}^{L}\left(1-\prod_{i=1}^{d}(w_{l,i}==y_{i}\ ||
\ w_{l,i}==b)\right)\right] \label{recall} \eerrr
Remember that $(w_{l,i}==y_{i}\ ||\ w_{l,i}==b)=0,1$ hence the two
rightmost products are just logical ands, and $1-()$ is a logical not.

Our algorithm is almost a bucketing forest, except that
the leaf condition is data dependent (for robustness).
A truly variable scheme can shape the buckets in a more complicated data
dependent way, see for example Gennaro Savino and Zezula \cite{GSZ01}.
Non-tree bucketing can use several coordinates together, so that the
resulting buckets are not boxes, see for example Andoni and Indyk \cite{AI06}
or \cite{MD2}.

\section{A Bucketing Forest Upper Bound} \label{bucketingupper}

In this section we will bound the performance of bucketing forest algorithms.
It is tricky, but technically simpler and more elegant than proving a
lower bound on the performance of a single algorithm.

\begin{theorem}
Assume the independent data model.
The success probability $P$ of a nonempty bucketing tree whose leafs
all have probabilities at most $1/N$ is at most
\beq P\le N^{-\lambda}\prod_{i=1}^{d}\max\left(1,
\sum_{j=0}^{b-1}\frac{p_{i,j}}{p_{i,j*}^{\lambda}}\right) \eeq
for any $\lambda\ge 0$. We do not even have to assume $p_{i,j}\le p_{i,j*}$.
\end{theorem}
\begin{proof}
Use induction. Without losing generality split coordinate $1$.
The induction step
\beq P\le\sum_{j=0}^{b-1}p_{1,j}\left(N p_{1,j*}\right)
^{-\lambda}\prod_{i=2}^{d}\max\left(1,
\sum_{j=0}^{b-1}\frac{p_{i,j}}{p_{i,j*}^{\lambda}}\right) \eeq
is valid for both proper and point-only subtrees. The maximization with
1 is necessary because coordinates can be ignored.
\end{proof}

\begin{theorem} \label{treeup}
Assume the independent data model.
Suppose an bucketing forest contains $T$ trees, its success probability is $S$,
and all its leafs have probabilities at most $1/N$. Than for any
$\lambda\ge 0$
\beqq \ln T\ge \lambda\ln N+\ln\frac{S}{2}-\sqrt{\frac{4}{S}\sum_{i=1}^{d}
V(P_{i},\lambda)}-\sum_{i=1}^{d}F(P_{i},\lambda) \eeqq
where
\beq V(P_{i},\lambda)=\sum_{j=0}^{b}p_{i,j}\left(\ln\frac{p_{i,j}}{q_{i,j}}
-\sum_{k=0}^{b}p_{i,k}\ln\frac{p_{i,k}}{q_{i,k}}\right)^{2} \eeq
and the $q_{i,j}$'s are the minimizing arguments from $F$'s definition
(\ref{up})
\end{theorem}
\begin{proof}
The previous theorem provides a good bound for the success probability of a
single tree, but it is not tight for a forest, because of
dependence: the failure of each tree increases the failure probability of
other trees. Now comes an interesting argument.
Recall that the success probability of the whole forest formula (\ref{recall}).
For any $z$ and $q_{i,j}>0$ we can bound
\beq S\le {\rm Prob}\{Z\ge z\} + e^{z}S_{Q} \eeq
where
\beq Z=\sum_{i=1}^{d}\ln\frac{p_{i,y_{i}}}{q_{i,y_{i}}} \eeq
\beqq {\rm Prob}\{Z\ge z\} =
\sum_{y\in\{0,1,\ldots,b\}^{d}}\prod_{i=1}^{d}p_{i,y_{i}} \cdot \left(
\sum_{i=1}^{d}\ln\frac{p_{i,y_{i}}}{q_{i,y_{i}}}\ge z\right) \eeqq
\berrr S_{Q} = \sum_{y\in\{0,1,\ldots,b\}^{d}}\prod_{i=1}^{d}q_{i,y_{i}}\cdot
\qquad\qquad\qquad\qquad\qquad \\ \qquad\qquad
\cdot\left[1-\prod_{l=1}^{L}\left(1-\prod_{i=1}^{d}(w_{l,i}==y_{i}\ ||
\ w_{l,i}==b) \right)\right] \eerrr
We insist upon
\beq \sum_{j=0}^{b}q_{i,j}=1 \eeq
so that we can use the previous lemma to bound
\beq S_{Q}\le TP_{q}\le T N^{-\lambda}\prod_{i=1}^{d}\max\left(1,
\sum_{j=0}^{b-1}\frac{q_{i,j}}{p_{i,j*}^{\lambda}}\right) \eeq
The other term is handled by the Chebyshev bound: for $z>{\rm E}(Z)$
\beq {\rm Prob}\{Z\ge z\} \le \frac{{\rm Var}(Z)}
{\left(z-{\rm E}(Z)\right)^{2}} \eeq
Together
\beq S\le \frac{{\rm Var}(Z)}{\left(z-{\rm E}(Z)\right)^{2}}+e^{z}S_{Q} \eeq
The reasonable choice of
\beq z={\rm E}(Z)+\sqrt{2{\rm Var}(Z)/S} \eeq
results in
\beq S\le 2e^{{\rm E}(Z)+\sqrt{2{\rm Var}(Z)/S}}S_{Q} \eeq
\end{proof}

Notice that this proof gives no indication that the bound is tight,
nor guidance towards constructing an actual bucketing forest, 
(except for telling which coordinates to throw away).

We tried to strengthen the theorem in the following way.
Instead of restricting the expected number of points falling into each
leaf bucket, allow larger leafs and only insist that
the total number of comparisons is at most $aN$. Surprisingly the
strengthened statement is wrong, and a 'large leafs' bucketing forest is
theoretically better than our algorithm. But it is complicated and non-robust.

\section{A Semi-Lexicographic Forest Upper Bound} \label{semiupper}

There remains the problem that we gave a lexicographic forest algorithm,
but a bucketing forest upper bound. It is a technicality,
which may be skipped over with little loss.
Any semi-lexicographic complete tree can be converted into a bucketing tree in
an obvious way: Prune the complete tree from the leafs down as much as
possible, preserving the property that each leaf expects at most $a/2$ points
from $X_{0}$. The success probability of the semi-lexicographic tree is
bounded by
\beq P\le P_{tree}+R \eeq
where $P_{tree}$ is the success probability of the truncated tree, for which we
have a good bound, and a remainder term associated with truncated tree vertexes
expecting more than $a/2$ tree points.

\begin{lemma}
Assume the independent data model and consider a semi-lexicographic tree
with the standard coordinate order (that does not restrict generality) and a
totally random values order. Assume that the special points pair
agree in coordinates $1,2,\ldots ,i-1$, but disagree at coordinate $i$ :
\beq y_{1},y_{2},\ldots,y_{i-1}\ne b,\ y_{i}=b \eeq
Conditioning on that, the probability of success is at most
\beq \frac{2a}{n_{0}p_{1,y{1}}p_{2,y_{2}}\cdots p_{i-1,y_{i-1}}} \eeq
\end{lemma}
\begin{proof}
Denote
\beq p=p_{1,y{1}}p_{2,y_{2}}\cdots p_{i-1,y_{i-1}} \eeq
Let $m$ be the number of $X_{0}$ points agreeing with the special pair
in their first $i-1$ coordinates. Its probability distribution is
$1$+Bernoulli($p,n_{0}-1$). Let us consider these $m$ points ordered by the
algorithm. The rank of the special $X_{0}$ point can be $1,2,\ldots,m$
with equal probabilities. Those $m$ ordered points are broken up into up to
$b$ intervals according to the value of coordinate $i$. Where does the
special $X_{1}$ point fit in? It is in a different interval than the $X_{0}$
special point, but its location in that interval, and the order of intervals
is random. Hence the probability that the two special points
are at most $a+1$ apart is at most $2a/m$. This has to be averaged:
\beq \sum_{m=1}^{n}\left(\barr{cc}n-1\\m-1\earr\right)
p^{m-1}(1-p)^{n-m}\frac{2a}{m}=\frac{2a}{np} \eeq
\end{proof}

\begin{theorem}
Assume the independent data model. Then the success probability of any
semi-lexicographic tree with a totally random coordinate values order
is at most 
\beq P\le \frac{2\ln\left(e^{4.5}N\right)}{N^{\lambda}}
\prod_{i=1}^{d}\max\left(1,
\sum_{j=0}^{b-1}\frac{p_{i,j}}{p_{i,j*}^{\lambda}}\right) \eeq
for any $0\le\lambda\le 1$, where
\beq N=\max\left(1,\frac{2n_{0}}{a}\right) \eeq
\end{theorem}
\begin{proof}
Without restricting generality assume that the coordinate have the standard
order. We have established that
\beqq R\le\!\!\!\!\!\!\!\!\!\!\!\!
\sum_{\footnotesize\barr{ccc}0\le t\le d \\ 0\le w_{1},w_{2},\ldots,w_{t}<b\\
N\prod_{i=1}^{t}p_{i,w_{i}*}\ge 1 \earr}\!\!\!\!\!\!\!\!\!\!\!\!
\frac{4}{N}\prod_{i=1}^{t}\frac{p_{i,w_{i}}}{p_{i,w_{i}*}}
\cdot \left(1-\sum_{j=0}^{b-1}p_{t+1,j}\right) \eeqq
The negative terms can be shifted to the next $t$ :
\beqq R\le\frac{4}{N}+\!\!\!\!\!\!\!\!\!\!\!\!
\sum_{\footnotesize\barr{ccc}1\le t\le d \\
0\le w_{1},w_{2},\ldots,w_{t}<b\\ N\prod_{i=1}^{t}p_{i,w_{i}*}\ge 1 \earr}
\!\!\!\!\!\!\!\!\!\!\!\!
\frac{4}{N}\prod_{i=1}^{t}\frac{p_{i,w_{i}}}{p_{i,w_{i}*}}
\cdot \left(1-p_{t,w_{t}*}\right) \eeqq
Denote
\beqq \tilde{R}_{w_{1},\ldots,w_{s}}=\!\!\!\!\!\!\!\!\!\!\!\!
\sum_{\footnotesize\barr{ccc}s\le t\le d \\
0\le w_{s+1},w_{s+2},\ldots,w_{t}<b\\ N\prod_{i=1}^{t}p_{i,w_{i}*}\ge 1 \earr}
\!\!\!\!\!\!\!\!\!\!\!\!
\prod_{i=s+1}^{t}\frac{p_{i,w_{i}}}{p_{i,w_{i}*}}
\cdot \left(1-p_{t,w_{t}*}\right) \eeqq
We will prove by induction from the leafs down that
\berr \tilde{R}_{w_{1},w_{2},\ldots,w_{s}}\le
N_{w_{1},\ldots,w_{s}}^{1-\lambda}\ln\left(eN_{w_{1},\ldots,w_{s-1}}\right)
\cdot\\ \cdot\prod_{i=s+1}^{d}\max\left(1,
\sum_{j=0}^{b-1}\frac{p_{i,j}}{p_{i,j*}^{\lambda}}\right) \eerr
where
\beq N_{w_{1},\ldots,w_{s}}=N\prod_{i=1}^{s}p_{i,w_{i}*} \eeq
The induction step boils down to
\beqq \ln\left(eN_{w_{1},\ldots,w_{s-1}}\right)
\ge \left(1-p_{s,w_{s}*}\right)+
\ln\left(eN_{w_{1},\ldots,w_{s-1}}p_{s,w_{s}*}\right) \eeqq
which is obviously true.
\end{proof}

Theorem (\ref{treeup}) is converted into
\begin{theorem}
Assume the independent data model.
Suppose a semi-lexicographic forest with a totally random coordinate values
order contains $T$ trees, its success probability is $S$, and
\beq N=\max\left(1,\frac{2n_{0}}{a}\right) \eeq
Than for any $0\le\lambda\le 1$
\berr \ln T\ge \lambda\ln N-\ln\left[2\ln\left(e^{4.5}N\right)\right]+
\qquad\qquad \\+
\ln\frac{S}{2}-\sqrt{\frac{4}{S}\sum_{i=1}^{d}
V(P_{i},\lambda)}-\sum_{i=1}^{d}F(P_{i},\lambda) \eerr
\end{theorem}

\section{A Lower Bound} \label{lower}

\begin{theorem} \label{law}
Assume the independent data model and denote
\beq N=\frac{2n_{0}}{a} \eeq
Let $\epsilon>0$ be some small parameter, and let
Let $\lambda,r_{1},r_{2},\ldots,r_{d}$ attain
\berr \min_{\lambda\ge 0}\max_{0\le r_{1},\ldots,r_{d}\le 1}\Bigg[
-(1+\epsilon)\lambda\ln N+\qquad\qquad\\ \qquad\qquad+
\sum_{i=1}^{d}\sum_{j=0}^{b}p_{i,j}\left(1-r_{i}+
r_{i}\frac{(j\ne b)}{p_{i,j*}^{\lambda}}\right)\Bigg] \eerr
The extrema conditions are
\beq \sum_{i=1}^{d}\sum_{j=0}^{b-1}p_{i,j}\frac
{-r_{i}\ln p_{i,j*}} {(1-r_{i})p_{i,j*}^{\lambda}+r_{i}}=(1+\epsilon)\ln N \eeq
and $r_{i}=0$ or
\beq \sum_{j=0}^{b-1}\frac{p_{i,j}}{(1-r_{i})p_{i,j*}^{\lambda}+r_{i}}=1
\qquad 1\le i\le d \eeq
Suppose that for some $\delta<1/7$
\berrr \sum_{i=1}^{d}\sum_{j=0}^{b}p_{i,j}\Bigg(
\ln[1-r_{i}+(j\ne b)r_{i}p_{i,j*}^{-\lambda}]-\qquad\qquad\qquad\\ \quad
-\sum_{k=0}^{b}p_{i,k}
\ln[1-r_{i}+(k\ne b)r_{i}p_{i,k*}^{-\lambda}]\Bigg)^{2}
\le \epsilon^{2}\delta\lambda^{2}\left(\ln N\right)^{2} \eerrr
\berrr \sum_{i=1}^{d}\sum_{j=0}^{b-1}p_{i,j}\Bigg(
\frac{-r_{i}\ln p_{i,j*}} {(1-r_{i})p_{i,j*}^{\lambda}+r_{i}}-
\qquad\qquad\qquad\qquad\\
-\sum_{k=0}^{b-1}p_{i,k}\frac{-r_{i}\ln p_{i,k*}}
{(1-r_{i})p_{i,k*}^{\lambda}+r_{i}}\Bigg)^{2}
\le \epsilon^{2}\delta\left(\ln N\right)^{2}/4 \eerrr
\beq \sum_{i=1}^{d}\sum_{j=0}^{b-1}p_{i,j}
\frac{r_{i}(1-r_{i})[\ln p_{i,j*}]^{2}}
{\left[(1-r_{i})p_{i,j*}^{\lambda}+r_{i}\right]^{2}}
\le \epsilon^{2}\delta\left(\ln N\right)^{2}/8 \eeq
Then the general algorithm with $T$ tries where
\berr \ln T\ge\ln\frac{1}{\delta}+(1+3\epsilon)\lambda\ln N-
\qquad\qquad\\ \qquad\qquad -
\sum_{i=1}^{d}\sum_{j=0}^{b}p_{i,j}\left(1-r_{i}+
r_{i}\frac{(j\ne b)}{p_{i,j*}^{\lambda}}\right) \eerr
has success probability
\beq S\ge 1-7\delta \eeq
Moreover there exists a bucketing forest with $T$ trees and at least
$1-7\delta$ success probability.
\end{theorem}

The alarmingly complicated small variance conditions are asymptotically
valid, because the variances grow linearly with $\ln N$.
However there is no guarantee that they can be always met.
Indeed the upper bound is of the Chernof inequality
large deviation type, and can be a poor estimate in pathological cases.
\begin{definition}
Let $Y,Z$ be joint random variables.
We denote by $Y_{Z}$ the conditional type random variable $Y$ with its
probability density multiplied by
\beq \frac{e^{Z}}{{\rm E}[e^{Z}]} \eeq
In the discrete case $Z,Y$ would have values $y_{i},z_{i}$ with
probability $p_{i}$. Then $Y_{Z}$ has values $y_{i}$ with probability
\beq \frac{p_{i}e^{z_{i}}}{\sum_{j}p_{j}e^{z_{j}}} \eeq
\end{definition}

\begin{lemma}
For any random variable $Z$, and $\lambda\ge 0$
\beq \ln{\rm Prob}\left\{Z\ge {\rm E}\left[Z_{\lambda Z}\right]\right\} \le
\ln{\rm E}\left[e^{\lambda Z}\right]-\lambda{\rm E}\left[Z_{\lambda Z}\right]
\eeq
\berr \ln{\rm Prob} \left\{Z\ge {\rm E}\left[Z_{\lambda Z}\right]-
\sqrt{2{\rm Var}\left[Z_{\lambda Z}\right]}\right\} \ge \\ \ge
\ln{\rm E}\left[e^{\lambda Z}\right]-
\lambda{\rm E}\left[Z_{\lambda Z}\right]-\ln 2-
\lambda\sqrt{2{\rm Var}\left[Z_{\lambda Z}\right]}  \eerr
\end{lemma}
\begin{proof}
The upper bound is the Chernof bound.
The lower bound combines the Chebyshev inequality
\beq {\rm Prob}\left\{\left|Z_{\lambda Z}-{\rm E}[Z_{\lambda Z}] \right| \le 
\sqrt{2{\rm Var}[Z_{\lambda Z}]}\right\}\ge \frac{1}{2} \eeq
with the fact that the condition in the curly bracket bounds the 
densities ratio:
\berr \ln\frac{e^{\lambda Z}}{{\rm E}\left[e^{\lambda Z}\right]}=
\ln\frac{e^{\lambda Z_{\lambda Z}}}{{\rm E}\left[e^{\lambda Z}\right]}\le\\ \le
-\ln{\rm E}\left[e^{\lambda Z}\right]+\lambda{\rm E}\left[Z_{\lambda Z}\right]+
\lambda\sqrt{2{\rm Var}\left[Z_{\lambda Z}\right]}  \eerr
\end{proof}

It is amusing, and sometimes useful to note that
\beq {\rm E}[Z_{\lambda Z}]=\frac{\partial\ln{\rm E}\left[e^{\lambda Z}\right]}
{\partial\lambda} \eeq
\beq {\rm Var}[Z_{\lambda Z}]=\frac{\partial^{2}\ln{\rm E}[e^{\lambda Z}]}
{\partial\lambda^{2}} \eeq
We will now prove the theorem \ref{law}.
\begin{proof}
Let $\lambda\ge 0$ be a parameter to be optimized.
Let $w\in\{0,1\}^{d}$ be the random Bernoulli vector
\beq w_{i}=(\lambda_{i}\le\lambda) \eeq
where $\lambda_{i}$ is the $i$'th random exponent.
In a slight abuse of notation let $0\le r_{i}\le 1$ denote not a random
variable but a probability
\beq r_{i}={\rm Prob}\{w_{i}==1\}={\rm Prob}\{\lambda_{i}\le\lambda\} \eeq
We could not resist doing that because
equation (\ref{imp}) is still valid under this interpretation.
Another point of view is to forget (\ref{imp}) and consider $r_{i}$
a parameter to be optimized.
Again let $y\in\{0,1,\ldots,b\}^{d}$ denote the abbreviated state
of the special points $x_{0},x_{1}$. Let us consider a single try of our
algorithm, conditioned on both $y$ and $w$. The following requirements
\beq \prod_{i=1}^{d}(1-w_{i}+w_{i}(y_{i}\ne b)) = 1 \label{eq} \eeq
\beq \prod_{i=1}^{d}(1-w_{i}+w_{i}p_{i,y_{i}*}) \le \frac{1}{N} =
\frac{a}{2n_{0}} \label{ineq} \eeq
state that the expected number of $X_{0}$ points in the bucket defined
by the coordinates whose $w_{i}=1$ with value $y_{i}$ is at most $a/2$.
Then the probability that the actual number of bucket points is more than
$a$ is bounded from above by $1/2$. A more compact way of stating
(\ref{eq}) and (\ref{ineq}) together is
\beq Z(y,w)\ge\ln N \eeq
\beq Z(y,w)=\sum_{i=1}^{d}\ln\left[1-w_{i}+w_{i}(y_{i}\ne b)p_{i,y_{i}*}^{-1}
\right] \eeq
Summing over $w$ gives success probability
of a single try, conditioned over $y$ to be at least
\beqq P(y)\ge\frac{1}{2}\sum_{w\in\{0,1\}^{d}}\prod_{i=1}^{d}
[(1-w_{i})(1-r_{i})+w_{i}r_{i}][Z(y,w)\ge\ln N] \eeqq
In short
\beq P(y)\ge\frac{1}{2}{\rm Prob}\left\{Z(y)\ge\ln N\right\}\label{sho}\eeq
Conditioning over $y$ makes tries independent of each other, hence the
conditional success probability of at least $T$ tries is at least
\beq S(y)\ge 1-\left(1-P(y)\right)^{T}\ge\frac{TP(y)}{1+TP(y)} \eeq
Averaging over $y$ bounds the success probability $S$ of the algorithm by
\beq S\ge \sum_{y\in\{0,1,\ldots,b\}^{d}}\prod_{i=1}^{d}p_{i,y_{i}}
\cdot\left[\frac{TP(y)}{1+TP(y)}\right] \eeq
In short
\beq S\ge{\rm E}\left[\frac{TP(y)}{1+TP(y)}\right] \label{short}\eeq

Now we must get our hands dirty. The reverse Chernof inequality is
\berrr \ln{\rm Prob} \bigg\{Z(y)\ge {\rm E}\left[Z(y)_{\lambda Z(y)}\right]-
\sqrt{2{\rm Var}\left[Z(y)_{\lambda Z(y)}\right]}\bigg\} \ge \\ \ge
\ln{\rm E}\left[e^{\lambda Z(y)}\right]-
\lambda{\rm E}\left[Z(y)_{\lambda Z(y)}\right]-\ln 2-\\-
\lambda\sqrt{2{\rm Var}\left[Z(y)_{\lambda Z(y)}\right]}  \eerrr
Denoting
\beqq U(y)=\ln{\rm E}\left[e^{\lambda Z(y)}\right]=
\sum_{i=1}^{d}\ln[1-r_{i}+(y_{i}\ne b)r_{i}p_{i,y_{i}*}^{-\lambda}] \eeqq
\berr V(y)=\frac{\partial U(y)}{\partial\lambda}=
{\rm E}\left[Z(y)_{\lambda Z(y)}\right]=\\=
\sum_{\footnotesize\barr{cc}1\le i\le d\\ y_{i}\ne b \earr}\frac
{-r_{i}\ln p_{i,y_{i}*}} {(1-r_{i})p_{i,y_{i}*}^{\lambda}+r_{i}}
\label{de}\eerr
\berr W(y)=\frac{\partial^{2} U(y)}{\partial\lambda^{2}}=
{\rm Var}\left[Z(y)_{\lambda Z(y)}\right]=\\=
\sum_{\footnotesize\barr{cc}1\le i\le d\\ y_{i}\ne b \earr}\frac
{r_{i}(1-r_{i})[\ln p_{i,y_{i}*}]^{2}}
{\left[(1-r_{i})p_{i,y_{i}*}^{\lambda}+r_{i}\right]^{2}} \eerr
the reverse Chernof inequality can be rewritten as
\berr \ln{\rm Prob} \left\{Z(y)\ge V(y)-\sqrt{2W(y)}\right\} \ge\\ \ge
U(y)-\lambda V(y)-\ln 2-\lambda\sqrt{2W(y)} \eerr
It is time for the second inequality tier. For any $\delta<1/3$
\berr {\rm Prob} \Big\{ |U(y)-{\rm E}[U]| \le \sqrt{{\rm Var}[U]/\delta},\\
|V(y)-{\rm E}[V]| \le \sqrt{{\rm Var}[V]/\delta},\\ 
W(y)\le {\rm E}[W]/\delta \Big\} \ge 1-3\delta \eerr
where
\beq {\rm E}[U]=\sum_{i=1}^{d}\sum_{j=0}^{b}
p_{i,j}\ln[1-r_{i}+(j\ne b)r_{i}p_{i,j*}^{-\lambda}] \eeq
\beq {\rm E}[V]=\sum_{i=1}^{d}\sum_{j=0}^{b-1}p_{i,j}\frac
{-r_{i}\ln p_{i,j*}} {(1-r_{i})p_{i,j*}^{\lambda}+r_{i}} \eeq
Hence
\berrr \ln{\rm Prob} \bigg\{Z(y)\ge {\rm E}[V]-\sqrt{{\rm Var}[V]/\delta}-
\sqrt{2{\rm E}[W]}/\delta\bigg\} \ge \\ \ge
{\rm E}[U]-\lambda {\rm E}[V]-\ln 2 -
\sqrt{{\rm Var}[U]/\delta}-\\-\lambda\sqrt{{\rm Var}[V]/\delta}-
\lambda\sqrt{2{\rm E}[W]/\delta} \eerrr
Now we have to pull all strings together. 
In order to connect with (\ref{sho}) we will require
\beq {\rm E}[V]=(1+\epsilon)\ln N \label{cond}\eeq
\beq \sqrt{{\rm Var}[V]}+\sqrt{2{\rm E}[W]}\le
\epsilon \delta^{1/2}\ln N \label{can}\eeq
for some small $\epsilon>0$. Recalling (\ref{de}), condition (\ref{cond})
is achieved by choosing $\lambda$ to attain
\beq \min_{\lambda\ge 0}\left[-(1+\epsilon)\lambda\ln N+{\rm E}[U]\right] \eeq
If (\ref{can}) holds, then
\beq \ln P(y)\ge -(1+2\epsilon)\lambda\ln N+{\rm E}[U]-\ln 4-
\sqrt{{\rm Var}[U]/\delta} \eeq
with probability at least $1-3\delta$.
Recalling (\ref{short}) the success probability is at least
\beq S\ge\frac{1-3\delta}{1+4e^{(1+2\epsilon)\lambda\ln N-{\rm E}[U]+
\sqrt{{\rm Var}[U]/\delta}}/T} \eeq
\end{proof}

\section{Conclusion}

To sum up, we present three things:
\begin{enumerate}
\item
An approximate nearest neighbor algorithm (\ref{imp}), and
its sparse approximation (\ref{cons}).
\item
An information style performance estimate (\ref{cut}).
\item
A warning against dimensional reduction of sparse data, see section
\ref{downside}.
\end{enumerate}




%

\end{document}